\documentclass[letterpaper, 10 pt, conference]{ieeeconf}  

\IEEEoverridecommandlockouts                              
\usepackage{lmodern,enumerate,rotating,tensor,anyfontsize}
\usepackage{amsmath,amsfonts,amssymb,subfigure,mathtools,mathabx,blkarray,centernot}
\usepackage{tikz-cd}
\usetikzlibrary{shapes,shapes.geometric,arrows,fit,calc,positioning,automata}
\usepackage[normalem]{ulem}

\usepackage{stmaryrd}

\allowdisplaybreaks

\tikzset{
->, 
node distance=3cm, 
every state/.style={thick, fill=gray!10}, 
initial text=$ $, 
}

\usetikzlibrary{decorations.pathmorphing}

\usepackage{times,amsmath,amssymb,graphicx,tikz,algorithm,algorithmic,url,hhline,booktabs,soul}
\usetikzlibrary{automata,arrows,positioning}
\usetikzlibrary{petri}

\makeatletter
\newcommand{\subalign}[1]{%
  \vcenter{%
    \Let@ \restore@math@cr \default@tag
    \baselineskip\fontdimen10 \scriptfont\tw@
    \advance\baselineskip\fontdimen12 \scriptfont\tw@
    \lineskip\thr@@\fontdimen8 \scriptfont\thr@@
    \lineskiplimit\lineskip
    \ialign{\hfil$\m@th\scriptstyle##$&$\m@th\scriptstyle{}##$\crcr
      #1\crcr
    }%
  }
}
\makeatother

\usepackage{algorithm}
\usepackage{algorithmic}
\usepackage{arydshln}

\usepackage{savesym}
\savesymbol{AND}
\savesymbol{OR}
\savesymbol{NOT}
\savesymbol{TO}
\savesymbol{COMMENT}
\savesymbol{BODY}
\savesymbol{IF}
\savesymbol{ELSE}
\savesymbol{ELSIF}
\savesymbol{FOR}
\savesymbol{WHILE}

\usepackage[colorlinks,urlcolor=blue,linkcolor=blue,citecolor=blue]{hyperref}

\definecolor{green}{rgb}{0.1,0.7,0.1}

\newcommand{\algorithmicstop}{\textbf{stop}}
\newcommand{\STOP}{\STATE \algorithmicstop}

\newtheorem{theorem}{Theorem}[section]

\newtheorem{corollary}{Corollary}[section]
\newtheorem{definition}{Definition}[section]
\newtheorem{example}{Example}[section]

\newtheorem{remark}{Remark}[section]

\newcounter{enumi_saved}

\title{\LARGE Verification of Strong $K$-Step Opacity for Discrete-Event Systems
}

\author{Xiaoguang Han, Kuize Zhang, and Zhiwu Li 
\thanks{This work was supported in part by the National Natural Science Foundation of China under Grant No. 61903274 and the Alexander von Humboldt Foundation. Submitted to IEEE CDC on March 28, 2022.
}
\thanks{Xiaoguang Han is with the College of Electronic Information and Automation, Tianjin University of Science and Technology, Tianjin 300222, China. {\tt\small hxg-allen@163.com.}
}
\thanks{Kuize Zhang is with Control Systems Group, Technical University of Berlin, Berlin 10587. {\tt\small kuize.zhang@campus.tu-berlin.de.}
}
\thanks{Zhiwu Li is with the Institute of Systems Engineering, Macau University of Science and Technology, Taipa 519020, China, and School of Electro-Mechanical Engineering, Xidian University, Xi'an 710071, China. {\tt\small zhwli@xidian.edu.cn.}
}
}

\begin{document}

\maketitle
\thispagestyle{empty}
\pagestyle{empty}

\begin{abstract}
  In this paper, we revisit the verification of strong $K$-step opacity ($K$-SSO) for partially-observed discrete-event systems modeled as nondeterministic finite-state automata.
  As a stronger version of the standard $K$-step opacity, $K$-SSO requires that an intruder cannot make sure whether or not a secret state has been visited within the last $K$ observable steps.
  To efficiently verify $K$-SSO, we propose a new concurrent-composition structure, which is a variant of our previously-proposed one.
  Based on this new structure, we design an algorithm for deciding $K$-SSO and prove that the proposed algorithm not only reduces the time complexity of the existing algorithms, but also does not depend on the value of $K$.
  Furthermore, a new upper bound on the value of $K$ in $K$-SSO is derived, which also reduces the existing upper bound on $K$ in the literature.
  Finally, we illustrate the proposed algorithm by a simple example.
\end{abstract}

\section{Introduction}\label{sec1}

Opacity is a concealment property, which requires that the secret information of a system cannot be distinguished from its non-secret information to a passive observer (called an intruder) who completely knows the system's structure but has only limited observations of its behavior.
In other words, an opaque system always holds the plausible deniability for its ``secrets" during its execution.
Opacity adapts to the characteristics of a variety of security and privacy requirements in diverse dynamic systems, including event-driven systems~\cite{Lafortune(2018),Hadjicostis(2020)}, time-driven systems~\cite{An(2020),Ramasubramanian(2020)}, and metric systems~\cite{Yin(2021)}.

The notion of opacity initially appeared in the computer science literature~\cite{Mazare(2004)} for analyzing cryptographic protocols.
Whereafter, various versions of opacity were introduced in the context of discrete-event systems (DES), including Petri nets~\cite{Bryans(2005)}, labelled transition systems~\cite{Bryans(2008)}, automata~\cite{Badouel(2007)}, etc.
For details see the recent surveys~\cite{Lafortune(2018),Jacob(2016)} and the textbook~\cite{Hadjicostis(2020)}.
Note that, in the literature the secrets of a system are modeled by two ways: 1) a set of secret states, and 2) a set of secret behaviors/traces.
For the former, opacity is referred to as state-based (e.g.,~\cite{Bryans(2005)}), while for the latter, opacity is referred to as language-based (e.g.,~\cite{Bryans(2008),Badouel(2007)}).

In automata-based formalisms, different notions of opacity were proposed in the literature, including current-state opacity (CSO)~\cite{Saboori(2007)}, initial-state opacity (ISO)~\cite{Saboori(2013)}, $K$-step opacity ($K$-SO)~\cite{Saboori(2011a)}, infinite-step opacity (Inf-SO)~\cite{Saboori(2012)}\footnote{For convenience, the notion originally named CSO (resp., ISO, $K$-SO, and Inf-SO) is categorized as standard CSO (resp., standard ISO, standard $K$-SO, and standard Inf-SO) in this paper.}, and language-based opacity (LBO)~\cite{Bryans(2008),Badouel(2007),Lin(2011)}. Some more efficient algorithms to check them have also been provided in~\cite{Wu(2013)}--\cite{Lan(2020)}.
In particular, it was proven that the above-mentioned five versions of opacity could be reduced to each other in polynomial time when LBO is restricted the special case that the secret languages are regular (cf.,~\cite{Wu(2013),Balun(2021),Balun(2022)}), while LBO is generally undecidable in finite-state automata with $\epsilon$-labeling functions (cf.,~\cite{Bryans(2008)}).
Furthermore, when a system is not opaque, a natural question to ask that ``how can one makes it opaque?
This is \emph{opacity enforcement problem}, which has been extensively investigated using a variety of techniques, including supervisory control~\cite{Dubreil(2010)}--\cite{Tong(2018)}, insertion or edit functions~\cite{Ji(2018)}--\cite{Yin(2020)}, dynamic observers~\cite{Zhang(2015)}, subobserver relationship~\cite{Moulton(2022)}, etc.
In addition, verification and/or enforcement of opacity have been extended to other classes of models, see, e.g.,~\cite{Tong(2017)}--\cite{Hou(2022)}.
Some applications of opacity in real-world systems have also been provided in the literature, see, e.g.,~\cite{Saboori(2011b)}--\cite{Lin(2020)}.

Among various notions of opacity, the standard CSO characterizes that an intruder cannot make sure whether a system is currently in a secret state.
In Location-Based Services (LBS), however, a user may want to hide his/her initial location or his/her location (e.g., visiting a hospital or bank) at some specific previous instant.
Such requirements can be characterized by the standard ISO and $K$/Inf-SO.
Note that, as mentioned in~\cite{Falcone(2015)}--\cite{Han(2022)}, these four standard versions of opacity have some limitations in practice.
Specifically, they cannot capture the situation that an intruder can never infer for sure whether a system has passed through a secret state based on his/her observations.
In other words, even though a system is ``\emph{opaque}" in the standard sense, the intruder may necessarily determine that a secret state must have been passed through.
To this end, in~\cite{Falcone(2015)}, a strong version of the standard $K$-SO called \emph{strong $K$-step opacity} ($K$-SSO) was proposed to capture that the visit of a secret state cannot be inferred within the last $K$ observable steps.
Inspired by~\cite{Falcone(2015)}, the notion of $K$-SSO was extended to \emph{strong infinite-step opacity} (Inf-SSO) in~\cite{Ma(2021)}, which is a strong version of the standard Inf-SO.
Accordingly, two algorithms have been provided to verify $K$-SSO and Inf-SSO using the so-called $K$-step recognizer and $\infty$-step recognizer, respectively.
In particular, the algorithm for verifying $K$-SSO reduces time complexity of that in~\cite{Falcone(2015)}.

Recently, in our previous work~\cite{Han(2022)}, two strong versions of the standard CSO and ISO, called \emph{strong current-state opacity} (SCSO) and \emph{strong initial-state opacity} (SISO), were proposed in nondeterministic finite-state automata, respectively.
Further, we developed a new methodology to simultaneously verify SCSO, SISO, and Inf-SSO using a concurrent-composition technique.
We also proved that the time complexity of the algorithm designed in~\cite{Han(2022)} for verifying Inf-SSO is lower than that in~\cite{Ma(2021)}.
Motivated by the results obtained in~\cite{Han(2022)}, in this paper, we proposed a new concurrent-composition structure that is a variant of that of~\cite{Han(2022)} to do more efficient verification for $K$-SSO compared with the results in~\cite{Falcone(2015),Ma(2021)}.
The main contributions of this paper are as follows.
\begin{itemize}
\item Although the concurrent-composition structure proposed in~\cite{Han(2022)} can determine SCSO, SISO, and Inf-SSO, it cannot be
  directly used to check $K$-SSO.
  To this end, we propose a new concurrent-composition structure to determine $K$-SSO, which is a variant of that of~\cite{Han(2022)}.
  Based on the proposed new structure, we design an improved algorithm for deciding $K$-SSO, which reduces the time complexity from $\mathcal{O}(|\Sigma_o||\Sigma_{uo}||X|^{2}2^{(K+2)|X|})$ proposed in~\cite{Ma(2021)} to
  $\mathcal{O}((|\Sigma_o||\Sigma_{uo}|+|\Sigma|)|X|^{2}2^{|X|})$, i.e., the proposed algorithm does not depend on the value of $K$.
\item Using the proposed concurrent-composition structure, a new upper bound on the value of $K$ in $K$-SSO is derived, i.e., $|\hat{X}|2^{|X\backslash X_S|}-1$.
  We prove that a system is Inf-SSO if and only if it is $(|\hat{X}|2^{|X\backslash X_S|}-1)$-SSO.
  This also reduces the existing upper bound $|X|(2^{|X|}-1)$ on $K$ derived in~\cite{Ma(2021)} when the size of a system is relatively large.
\end{itemize}

The rest of this paper is arranged as follows.
Section~\ref{sec2} provides some basic notions needed in this paper.
In Section~\ref{sec3}, a new concurrent-composition structure is proposed, based on which a more efficient verification algorithm for $K$-SSO is designed, as well as an improved upper bound on the value of $K$ in $K$-SSO.
Finally, we conclude this paper in Section~\ref{sec4}.

\section{Preliminaries}\label{sec2}

A \emph{nondeterministic finite-state automaton} (NFA) is a quadruple structure $G=(X,\Sigma,\delta, X_0)$, where $X$ is a finite set of states, $\Sigma$ is a finite set of events, $X_0\subseteq X$ is a set of initial states, $\delta: X\times\Sigma\rightarrow 2^{X}$ is the transition function, which depicts the system dynamics: given states $x,y\in X$ and an event $\sigma\in \Sigma$, $y\in\delta(x,\sigma)$ implies that there exists a transition labeled by $\sigma$ from $x$ to $y$.
We can extend the transition function to $\delta :X\times\Sigma^{\ast}\rightarrow 2^{X}$ in the recursive manner, where $\Sigma^{\ast}$ denotes the \emph{Kleene closure} of $\Sigma$, consisting of all finite sequences composed of the events in $\Sigma$ (including the empty sequence $\epsilon$). Details can be found in~\cite{Hopcroft(2001)}.
We use $\mathcal{L}(G,x)$ to denote the language generated by $G$ from state $x$, i.e., $\mathcal{L}(G,x)=\{s\in\Sigma^{\ast}: \delta(x,s)\neq\emptyset\}$.
Therefore, the language generated by $G$ is $\mathcal{L}(G)=\cup_{x_0\in X_0}\mathcal{L}(G,x_0)$.
$G$ is called a deterministic finite-state automation (DFA) if $|X_0|=1$ and $|\delta(x,\sigma)|\leq 1$ for all $x\in X$ and all $\sigma\in\Sigma$.
When $G$ is deterministic, $\delta$ is also regarded as a partial transition function $\delta: X\times\Sigma^{\ast}\rightarrow X$.
For a sequence $s\in\mathcal{L}(G)$, we denote its length by $|s|$ and its prefix closure by $Pr(s)$, i.e., $Pr(s)=\{w\in\mathcal{L}(G): (\exists w^\prime\in\Sigma^{\ast})[ww^\prime=s]\}$.
Further, for a prefix $w\in Pr(s)$, we use the notation $s/w$ to denote the suffix of $s$ after its prefix $w$.

In this paper, a DES of interest is modeled as an NFA $G$.
As usual, we assume that the intruder can only see partially the behavior of $G$.
To this end, $\Sigma$ is partitioned into the set $\Sigma_o$ of observable events and the set $\Sigma_{uo}$ of unobservable events, i.e., $\Sigma_o\cup\Sigma_{uo}=\Sigma$ and $\Sigma_o\cap\Sigma_{uo}=\emptyset$.
The natural projection $P:\Sigma^{\ast}\rightarrow \Sigma_o^{\ast}$ is defined recursively by
($i$) $P(\epsilon)=\epsilon$,
($ii$) $P(s\sigma)=P(s)\sigma,\mbox{ if } \sigma\in \Sigma_o$,
and ($iii$) $P(s\sigma)=P(s),\mbox{ if } \sigma\in \Sigma_{uo}$, where $s\in\Sigma^\ast$.
We extend the natural projection $P$ to $\mathcal{L}(G)$ by $P(\mathcal{L}(G))=\{P(s)\in\Sigma_{o}^{\ast}: s\in\mathcal{L}(G)\}$, see, e.g.,~\cite{Cassandras(2008)} for details.
Without loss of generality, we assume that system $G$ is accessible, i.e., all its states are reachable from $X_0$.
A state $x\in X$ is called \emph{$K$-step observationally reachable} if there exists an initial state $x_0\in X_0$ and a sequence $s\in\mathcal{L}(G,x_0)$ such that
$x\in\delta(x_0,s)$ and $|P(s)|=K$, where $K\in\mathbb{N}$ is a natural number.

To study the verification of strong $K$-step opacity of $G=(X,\Sigma,\delta,X_0)$, we assume that $G$ has a set of secret states, denoted by $X_{S}\subseteq X$.
Then, $X_{NS}=X\backslash X_S$ is the set of non-secret states.
Consider an $n$-length sequence $s=s_{1}s_{2}\ldots s_{n}\in\Sigma^{\ast}$, $x_0\in X_0$, and $x_i\in X$, $i=1,2,\ldots,n$,
if $x_{k+1}\in\delta(x_k,s_{k+1})$, $0\leq k\leq n-1$, we call $x_0\stackrel{s_1}{\rightarrow}x_1\stackrel{s_2}{\rightarrow}x_2\stackrel{s_3}{\rightarrow}\cdots\stackrel{s_n}{\rightarrow}x_n$ a \emph{run} generated by $G$ from $x_0$ to $x_n$ under $s$.
For brevity, we write $x_0\stackrel{s}{\rightarrow}x_n$ (resp., $x_0\stackrel{s}{\rightarrow}$) when $x_1,x_2,\ldots,x_{n-1}$ (resp., $x_1,x_2,\ldots,x_n$) are not specified.
Note that $x_0\stackrel{s}{\rightarrow}x_n$ (resp., $x_0\stackrel{s}{\rightarrow}$) may denote more than one run based on the nondeterminism of $G$, which depends on the context.
A run $x_0\stackrel{s_1}{\rightarrow}x_1\stackrel{s_2}{\rightarrow}x_2\stackrel{s_3}{\rightarrow}\cdots\stackrel{s_n}{\rightarrow}x_n$ (resp., $x_0\stackrel{s_1}{\rightarrow}x_1\stackrel{s_2}{\rightarrow}x_2\stackrel{s_3}{\rightarrow}\cdots$), abbreviated as $x_0\stackrel{s}{\rightarrow}x_n$ (resp., $x_0\stackrel{s}{\rightarrow}$), is called \emph{non-secret} if $x_i\in X_{NS}$, $i=0,1,2,\cdots,n$ (resp., $i=0,1,2,\cdots$).

\section{Verification of strong $K$-step opacity}\label{sec3}

\subsection{Notion of strong $K$-step opacity }\label{subsec3.1}

In~\cite{Falcone(2015)}, the authors proposed a notion of strong $K$-step opacity ($K$-SSO) for a DFA $G$.
Specifically, $G$ is said to be strongly $K$-step opaque ($K$-SSO)\footnote{In this paper, the terminology ``$K$-SSO" is the acronym of both ``strong $K$-step opacity" and ``strongly $K$-step opaque", which depends on the context.} w.r.t. $\Sigma_o$ and $X_{S}$ if for all $st\in\mathcal{L}(G,x_0)$ such that $\delta(x_0,s)\in X_S$ and $|P(t)|\leq K$,
there exists $w\in\mathcal L(G,x_0)$ such that $P(w)=P(st)$ and for all $\bar{w}\in Pr(w)$, if $|P(w/\bar{w})|\leq K$, then $\delta(x_0,\bar{w})\notin X_S$.
In this subsection, we reformulate the definition of $K$-SSO in nondeterministic finite-state automata.
And then, we do $K$-SSO verification using the proposed concurrent-composition structure, which reduces the (worst-case) time complexity of the previous algorithms in~\cite{Falcone(2015),Ma(2021)}.

\begin{definition}[$K$-SSO]\label{de:3.1}
Given a system $G=(X,\Sigma,\delta,$ $X_0)$, a projection map $P$ w.r.t. the set $\Sigma_o$ of observable events, and a set $X_{S}\subseteq X$ of secret states,
$G$ is said to be strongly $K$-step opaque ($K$-SSO) w.r.t. $\Sigma_o$ and $X_{S}$ (where $K\in\mathbb{N}$) if
\begin{align}\label{eq:3.1}
&(\forall\mbox{ run } x_0\stackrel{s}{\rightarrow}x_s\stackrel{t}{\rightarrow}x_t: x_0\in X_0\wedge x_s\in X_s\wedge|P(t)|\leq K)\nonumber\\
&(\exists\mbox{ run }x^\prime_0\stackrel{s^\prime}{\rightarrow}{x^\prime_s}\stackrel{t^\prime}{\rightarrow}{x^\prime_t})[(x^\prime_0\in X_0)\wedge(P(s^\prime)=P(s))\wedge\nonumber\\
&(P(t^\prime)=P(t))\wedge({x^\prime_s}\stackrel{t^\prime}{\rightarrow}{x^\prime_t}\mbox{ is non-secret})].
\end{align}
\end{definition}

\begin{remark}\label{re:3.1}
Obviously, $K$-SSO in Definition~\ref{de:3.1} is more general than the notion in~\cite{Falcone(2015)} or~\cite{Ma(2021)}.
In plain words, if a system is $K$-SSO, then an external intruder cannot make sure whether the system is/was in a secret state within the last $K$ observable steps.
Compared with the standard $K$-SO proposed in~\cite{Saboori(2011a)}, $K$-SSO has a higher-level confidentiality.
In other words, $K$-SSO implies the standard $K$-SO, but the converse is not true.
\end{remark}

\subsection{Structure of concurrent composition}\label{subsec3.2}

In this subsection, we propose a new information structure using a concurrent-composition approach to verify $K$-SSO in Definition~\ref{de:3.1}.
Note that, the proposed information structure is a variant of that proposed in~\cite{Han(2022)}.
Later on, we will show that the time complexity of using the proposed concurrent-composition structure to verify $K$-SSO is lower than those in~\cite{Falcone(2015),Ma(2021)} and our proposed algorithm does not depend on the value of $K$.
In order to present this structure, we need to introduce the notions of \emph{initial-secret subautomaton} and \emph{non-secret subautomaton}.

Given a system $G=(X,\Sigma,\delta,X_0)$ and a set $X_S\subseteq X$ of secret states.
We first recall the notion of standard \emph{subset construction} of $G$ called an \emph{observer}, which is defined by
\begin{equation}\label{eq:3.2}
Obs(G)=(X_{obs},\Sigma_{obs},\delta_{obs},X_{obs,0}),
\end{equation}
where $X_{obs}\subseteq 2^X\backslash\{\emptyset\}$ stands for the set of states,
$\Sigma_{obs}=\Sigma_o$ stands for the set of observable events,
$\delta_{obs}: X_{obs}\times\Sigma_{obs}\rightarrow X_{obs}$ stands for the (partial) deterministic transition function defined as follows: for any $q\in X_{obs}$ and $\sigma\in\Sigma_{obs}$,
we have $\delta_{obs}(q,\sigma)=\{x^{\prime}\in X:\exists x\in q, \exists w\in\Sigma_{uo}^\ast\mbox{ s.t. }x^{\prime}\in\delta(x,\sigma w)\}$ if it is nonempty,
$X_{obs,0}=\{x\in X:\exists x_0\in X_0, \exists w\in\Sigma_{uo}^\ast\mbox{ s.t. }x\in\delta(x_0,w)\}$ stands for the (unique) initial state.
For brevity, we only consider the accessible part of observer $Obs(G)$.
We refer the reader to~\cite{Cassandras(2008)} for details on $Obs(G)$.

\begin{remark}\label{re:3.2}
According to Definition~\ref{de:3.1}, $K$-SSO reduces to the standard CSO when $K=0$.
Therefore, $0$-SSO can be determined using the observer $Obs(G)$.
Specifically, $G$ is $0$-SSO if and only if there exists no reachable state $q\in X_{obs}$ such that $q\subseteq X_S$, see~\cite{Saboori(2007)} for details.
Hence, from now on we always assume $K\geq 1$ when we use the following proposed concurrent-composition approach to determine $K$-SSO.
\end{remark}

To study verification of $K$-SSO ($K\geq 1$), the state set $X_{obs}$ of $Obs(G)$ is partitioned into three disjoint parts: $X^s_{obs}$, $X^{ns}_{obs}$, and $X^{hyb}_{obs}$, i.e., $X_{obs}=X^s_{obs}\cup X^{ns}_{obs}\cup X^{hyb}_{obs}$, where $X^s_{obs}=\{q\in X_{obs}:q\subseteq X_S\}$, $X^{ns}_{obs}=\{q\in X_{obs}:q\subseteq X_{NS}\}$, and $X^{hyb}_{obs}=\{q\in X_{obs}:q\cap X_S\neq\emptyset\wedge q\cap X_{NS}\neq\emptyset\}$.
Note that the superscript ``\emph{hyb}" of $X^{hyb}_{obs}$ stands for the acronym of ``\emph{hybrid}".

Now we construct a subautomaton of $G$ called an \emph{initial-secret subautomaton}, denoted by
\begin{equation}\label{eq:3.3}
\hat{G}=(\hat{X},\hat{\Sigma},\hat{\delta},\hat{X}_0),
\end{equation}
which is obtained from $G$ by: 1) replacing its initial state set $X_0$ with $\hat{X}_0=X_S$, and 2) computing the part of $G$ reachable from $\hat{X}_0$ as $\hat{G}$.
Note that $\hat{G}$ can be computed from $G$ in time linear in the size of $G$.
In particular, when $G$ is deterministic, the time complexity of computing $\hat{G}$ reduces to $\mathcal{O}(|\Sigma||X|)$.

We also construct another subautomaton of $G$ called a \emph{non-secret subautomaton}, denoted by
\begin{equation}\label{eq:3.4}
\tilde{G}=(\tilde{X},\tilde{\Sigma},\tilde{\delta},\tilde{X}_0),
\end{equation}
whose set of initial states is defined as $\tilde{X}_0=\{x\in X:\exists q\in X^{hyb}_{obs}\mbox{ s.t. }x\in q\cap X_{NS}\}$.
And then in $G$ we delete all secret states and compute the part of $G$ reachable from $\tilde{X}_0$ as $\tilde{G}$.
The time complexity of computing $\tilde{G}$ from $G$ is exponential in the size of $G$.
Note that when we delete a secret state, all transitions attached to that state are also deleted.

Next, we construct a new observer of $\tilde{G}$, denoted by $\tilde{G}_{obs}$, which is a minor variant of standard observer $Obs(\tilde{G})$.
The unique difference between them is the set of initial states.
Specifically, we construct $\tilde{G}_{obs}=(\tilde{X}_{obs},\tilde{\Sigma}_{obs},\tilde{\delta}_{obs},\tilde{X}_{obs,0})$, where the initial state set is $\tilde{X}_{obs,0}=\{X_{NS}\cap q:q\in X^{hyb}_{obs}\}$.

Based on the above preparation, we propose an information structure called the \emph{concurrent composition} of $\hat{G}$ and $\tilde{G}_{obs}$, which will be used to verify $K$-SSO in Definition~\ref{de:3.1}.

\begin{definition}[Concurrent Composition]\label{de:3.2}
Given a system $G=(X,\Sigma,\delta,X_0)$ and a set $X_{S}\subseteq X$ of secret states, the concurrent composition of $\hat{G}$ and $\tilde{G}_{obs}$ is an NFA
\begin{equation}\label{eq:5}
Cc(\hat{G},\tilde{G}_{obs})=(\hat{X}_{cc},\hat{\Sigma}_{cc},\hat{\delta}_{cc},\hat{X}_{cc,0}),
\end{equation}
where
\begin{itemize}
  \item $\hat{X}_{cc}\subseteq \hat{X}\times 2^{\tilde{X}}$ stands for the set of states;
  \item $\hat{\Sigma}_{cc}=\{(\sigma,\sigma): \sigma\in\hat{\Sigma}_o\}\cup\{(\sigma,\epsilon): \sigma\in\hat{\Sigma}_{uo}\}$ stands for the set of events;
  \item $\hat{\delta}_{cc}: \hat{X}_{cc}\times\hat{\Sigma}_{cc}\rightarrow 2^{\hat{X}_{cc}}$ is the transition function defined as follows: for any state $(x,q)\in\hat{X}_{cc}$ and any event $\sigma\in\hat{\Sigma}$,
  \begin{itemize}
  \item [(i)] when $q\neq\emptyset$,\\
    (a) if $\sigma\in\hat{\Sigma}_o$, then
   \begin{equation*}
   \begin{split}
  & \hat{\delta}_{cc}((x,q),(\sigma,\sigma))=\{(x^\prime,q^\prime): x^\prime\in\hat{\delta}(x,\sigma)\wedge\\
  & q^\prime=\tilde{\delta}_{obs}(q,\sigma)\mbox{ if }\tilde{\delta}_{obs}(q,\sigma) \mbox{ is well-defined}, q^\prime=\emptyset \\
  & \mbox{otherwise}\};
   \end{split}
   \end{equation*}
    (b) if $\sigma\in\hat{\Sigma}_{uo}$, then
   \begin{equation*}
   \hat{\delta}_{cc}((x,q),(\sigma,\epsilon))=\{(x^\prime,q): x^\prime\in\hat{\delta}(x,\sigma)\};
  \end{equation*}
  \item [(ii)] When $q=\emptyset$,\\
     (a) if $\sigma\in\hat{\Sigma}_o$, then
   \begin{equation*}
   \hat{\delta}_{cc}((x,\emptyset),(\sigma,\sigma))=\{(x^\prime,\emptyset): x^\prime\in\hat{\delta}(x,\sigma)\};
   \end{equation*}
     (b) if $\sigma\in\hat{\Sigma}_{uo}$, then
   \begin{equation*}
   \hat{\delta}_{cc}((x,\emptyset),(\sigma,\epsilon))=\{(x^\prime,\emptyset): x^\prime\in\hat{\delta}(x,\sigma)\};
  \end{equation*}
  \end{itemize}
  \item $\hat{X}_{cc,0}=\{(x,y): (\exists q\in X^{hyb}_{obs})\mbox{ s.t. }[(x\in X_S\cap q)\wedge(y=X_{NS}\cap q)]\}\subseteq\hat{X}_0\times\tilde{X}_{obs,0}$ stands for the set of initial states.
\end{itemize}
\end{definition}

\begin{remark}\label{re:3.3}
The concurrent composition $Cc(\hat{G},\tilde{G}_{obs})$ in Definition~\ref{de:3.2} is a variant of that of~\cite{Han(2022)}.
The key differences between them are as follows: 1) set of initial states, and 2) composite objects.
Intuitively, $Cc(\hat{G},\tilde{G}_{obs})$ captures that for all $x_s\in X_S$ and all $s\in\mathcal{L}(G,x_s)$ whether there exists an observation $\alpha\in\mathcal{L}(\tilde{G}_{obs}, y)$ such that $P(s)=\alpha$ and the initial state $y$ of $\tilde{G}_{obs}$ satisfies $y=q\backslash X_S$ for some $q$ satisfying $x_s\in q\in X^{hyb}_{obs}$.
In addition, for a sequence $e\in\mathcal{L}(Cc(\hat{G},\tilde{G}_{obs}))$, we use the notations $e(L)$ and $e(R)$ to denote its left and right components, respectively.
Further, $P(e)$ denotes $P(e(L))$ or $e(R)$ because $P(e(L))=P(e(R))=e(R)$, which depends on the context.
\end{remark}

\subsection{Verification for strong $K$-step opacity}\label{subsec3.3}

In this subsection, we are ready to present the main result on the verification of $K$-SSO using the proposed concurrent composition $Cc(\hat{G},\tilde{G}_{obs})$.

\begin{theorem}\label{th:3.1}
Given a system $G=(X,\Sigma,\delta,X_0)$, a projection map $P$ w.r.t. the set $\Sigma_o$ of observable events, and a set $X_{S}$ of secret states,
let $Cc(\hat{G},\tilde{G}_{obs})$ be the corresponding concurrent composition.
$G$ is $K$-SSO w.r.t. $\Sigma_o$ and $X_S$ with $K\geq 1$ if and only if there exists no state of the form $(\cdot,\emptyset)$ in $Cc(\hat{G},\tilde{G}_{obs})$ that is observationally reachable from $\hat{X}_{cc,0}$ within $K$-steps.
\end{theorem}

\begin{proof}
$(\Rightarrow)$ By contrapositive, assume that there exists a state $(x,\emptyset)$ in $Cc(\hat{G},\tilde{G}_{obs})$ that is $k$-step observationally reachable from $\hat{X}_{cc,0}$, where $k\leq K$.
By the construction of $Cc(\hat{G},\tilde{G}_{obs})$, we have that there exists an initial state $(x_s,y)\in\hat{X}_{cc,0}$ and a sequence $e\in\mathcal{L}(Cc(\hat{G},\tilde{G}_{obs}),(x_s,y))$ with $|P(e)|=k$ such that $(x,\emptyset)\in\hat{\delta}_{cc}((x_s,y),e)$, where $x_s\in X_S\cap q$, $y=X_{NS}\cap q$, and $q\in X^{hyb}_{obs}$.
Further, By the constructions of $\hat{G}$ and $\tilde{G}_{obs}$, we have that: 1) $x\in\hat{\delta}(x_s,e(L))$, and 2) $\tilde{\delta}_{obs}(y,e(R))$ is not well-defined.
Item 1) means $x\in\delta(x_s,e(L))$.
Item 2) means, by the constructions of $\tilde{G}$, that for all $x^\prime_s\in y$ and all $t^\prime\in\tilde{\Sigma}^\ast$ with $P(t^\prime)=e(R)$, it holds $\tilde{\delta}(x^\prime_s,t^\prime)=\emptyset$.
On the other hand, since $\{x_s\}\cup y\subseteq q\in X^{hyb}_{obs}$, by the construction of $Obs(G)$, we conclude that: 1) in $G$ there exists an initial state $x_0\in X_0$ and a sequence $s\in\mathcal{L}(G,x_0)$ such that $x_s\in\delta(x_0,s)$,
and 2) for each $x^\prime_s\in y$, there exists an initial state $x^\prime_0\in X_0$ and a sequence $s^\prime\in\mathcal{L}(G,x^\prime_0)$ with $P(s^\prime)=P(s)$ such that $x^\prime_s\in\delta(x^\prime_0,s^\prime)$.
Therefore, for the run $x_0\stackrel{s}{\rightarrow}x_s\stackrel{e(L)}{\rightarrow}x$ generated by $G$ with $|P(e(L))|=k\leq K$,
there exists no run $x^\prime_0\stackrel{s^\prime}{\rightarrow}{x^\prime_s}\stackrel{t^\prime}{\rightarrow}{x^\prime}$ with $P(s^\prime)=P(s)$ and $P(t^\prime)=P(t)$ such that its subrun ${x^\prime_s}\stackrel{t^\prime}{\rightarrow}{x^\prime_t}$ is non-secret.
By Definition~\ref{de:3.1}, $G$ is not $K$-SSO w.r.t. $\Sigma_o$ and $X_S$.

$(\Leftarrow)$ Also by contrapositive, assume that $G$ is not $K$-SSO w.r.t. $\Sigma_o$ and $X_S$.
By Definition~\ref{de:3.1}, we conclude that in $G$: 1) there exists a run $x_0\stackrel{s}{\rightarrow}x_s\stackrel{t}{\rightarrow}x_t$, where $x_0\in X_0$, $x_s\in X_S$, and $P(t)=k\leq K$,
and 2) there exists no run $x^\prime_0\stackrel{s^\prime}{\rightarrow}{x^\prime_s}\stackrel{t^\prime}{\rightarrow}{x^\prime_t}$ such that ${x^\prime_s}\stackrel{t^\prime}{\rightarrow}{x^\prime_t}$ is non-secret, where $x^\prime_0\in X_0$, $P(s^\prime)=P(s)$, and $P(t^\prime)=P(t)$.
By the construction of $\hat{G}$, item 1) means $x_t\in\hat{\delta}(x_s,t)$.
By the construction of $\tilde{G}$, item 2) means $\tilde{\delta}(x^{\prime}_s,t^\prime)=\emptyset$ for all $x^\prime_s\in y:=X_{NS}\cap q$, where $x_s\in q\in X^{hyb}_{obs}$, $t^\prime\in\tilde{\Sigma}^\ast$, and $P(t^\prime)=P(t)$.
By the construction of  $\tilde{G}_{obs}$, we further obtain $\tilde{\delta}_{obs}(y,P(t^\prime))=\emptyset$.
Then, by the construction of $Cc(\hat{G},\tilde{G}_{obs})$, we conclude that there exists a sequence $e\in\mathcal{L}(Cc(\hat{G},\tilde{G}_{obs}))$ with $e(L)=t$ and $e(R)=P(t^\prime)$ such that $(x_t,\emptyset)\in\hat{\delta}_{cc}((x_s,y),e)$.
Since $|P(e)|=|P(e(L))|=|P(t)|=k$, state $(x_t,\emptyset)$ is $k$-step observationally reachable from $\hat{X}_{cc,0}$ in $Cc(\hat{G},\tilde{G}_{obs})$.
\end{proof}

Based on Theorem~\ref{th:3.1}, a verification procedure for $K$-SSO can be summed up as the following Algorithm~\ref{algo:1}.
\begin{algorithm}
\caption{Verification of $K$-SSO}\label{algo:1}
\begin{algorithmic}[1]
\REQUIRE A system $G=(X,\Sigma,\delta,X_0)$, a set $\Sigma_o$ of observable events, and a set $X_{S}$ of secret states.
\ENSURE ``Yes'' if $G$ is $K$-SSO w.r.t. $\Sigma_o$ and $X_S$, ``No" otherwise.
\STATE Compute the observer $Obs(G)$ of $G$
\IF{there exists a reachable state $q\in X_{obs}$ such that $q\subseteq X_S$}
\STATE $G$ is not $0$-SSO w.r.t. $\Sigma_o$ and $X_S$, return ``No"
\STOP
\ELSE
\STATE Construct the initial-secret subautomaton $\hat{G}$ of $G$
\STATE Construct the non-secret subautomaton $\tilde{G}$ of $G$
\STATE Compute the observer $\tilde{G}_{obs}$ of $\tilde{G}$
\STATE Compute the corresponding $Cc(\hat{G},\tilde{G}_{obs})$
\STATE Use the ``Breadth-First Search Algorithm" in~\cite{Cormen(2009)} to find whether there exists a state of form $(\cdot,\emptyset)$ in $Cc(\hat{G},\tilde{G}_{obs})$ that is observationally reachable from $\hat{X}_{cc,0}$ within $K$-steps
\IF{such a state $(\cdot,\emptyset)$ in $Cc(\hat{G},\tilde{G}_{obs})$ exists}
\RETURN ``No", stop
\ELSE
\RETURN ``Yes", stop
\ENDIF
\ENDIF
\end{algorithmic}
\end{algorithm}

\begin{remark}\label{re:3.4}
We highlight the main advantages of using the proposed concurrent composition $Cc(\hat{G},\tilde{G}_{obs})$ to determine $K$-SSO compared with the existing algorithm in~\cite{Ma(2021)}.
First, computing $Obs(G)$, $\hat{G}$, $\tilde{G}$, $\tilde{G}_{obs}$, and $Cc(\hat{G},\tilde{G}_{obs})$ take time $\mathcal{O}(|\Sigma_o||\Sigma_{uo}||X|^{2}2^{|X|})$, $\mathcal{O}(|\Sigma||X|^2)$, $\mathcal{O}(|\Sigma||X|^2+2^{|X|})$,
$\mathcal{O}(|\Sigma_o||\Sigma_{uo}||X|^{2}2^{|X|})$, and $\mathcal{O}(|\Sigma||X|^{2}2^{|X|})$, respectively.
Hence, the overall (worst-case) time complexity of verifying $K$-SSO using Algorithm~\ref{algo:1} is $\mathcal{O}((|\Sigma_o||\Sigma_{uo}|+|\Sigma|)|X|^{2}2^{|X|})$.
In comparison, the algorithm in~\cite{Ma(2021)} using $K$-step recognizer has time complexity $\mathcal{O}(|\Sigma_o||\Sigma_{uo}||X|^{2}2^{(K+2)|X|})$.
Therefore, our proposed algorithm leads to a considerable improvement compared with that in~\cite{Ma(2021)}.
Second, by Theorem~\ref{th:3.1}, we know that the proposed concurrent composition $Cc(\hat{G},\tilde{G}_{obs})$ for determining $K$-SSO does not depend on the value of $K$, whereas the algorithm in~\cite{Ma(2021)} depends on the value of $K$.
In other words, for each given $K$, it needs to construct the corresponding recognizer, see~\cite{Ma(2021)} for details.
\end{remark}

\begin{example}[\cite{Ma(2021)}]\label{ex:3.1}
Let us consider the system $G$ shown in Fig.~\ref{Fig1} in which the set of secret states is $X_S=\{5,7\}$.
By applying Algorithm~\ref{algo:1}, we obtain the corresponding $Obs(G)$, $\hat{G}$, $\tilde{G}$, $\tilde{G}_{obs}$, and $Cc(\hat{G},\tilde{G}_{obs})$, which are depicted in~Fig.~\ref{Fig2}.
In $Cc(\hat{G},\tilde{G}_{obs})$ there exists a state $(8,\emptyset)$ that is $2$-step observationally reachable from the initial-state $(7,\{1,2,3,4\})$.
Therefore, by Theorem~\ref{th:3.1}, $G$ is $1$-SSO w.r.t. $\Sigma_o$ and $X_S$, but not $K$-SSO for any $K>1$.
This conclusion coincides with that obtained in~\cite{Ma(2021)}.
\begin{figure}[!ht]
  \centering
  \includegraphics[scale=0.76]{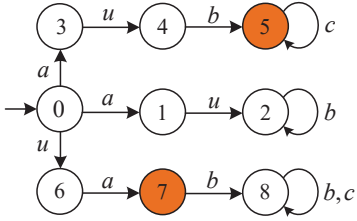}
  \caption{The system $G$ considered in Example~\ref{ex:3.1}, where $\Sigma_o=\{a,b,c\}$, $\Sigma_{uo}=\{u\}$, and $X_0=\{0\}$.}
  \label{Fig1}
\end{figure}

\begin{figure}[!ht]
  \centering
  \includegraphics[scale=0.92]{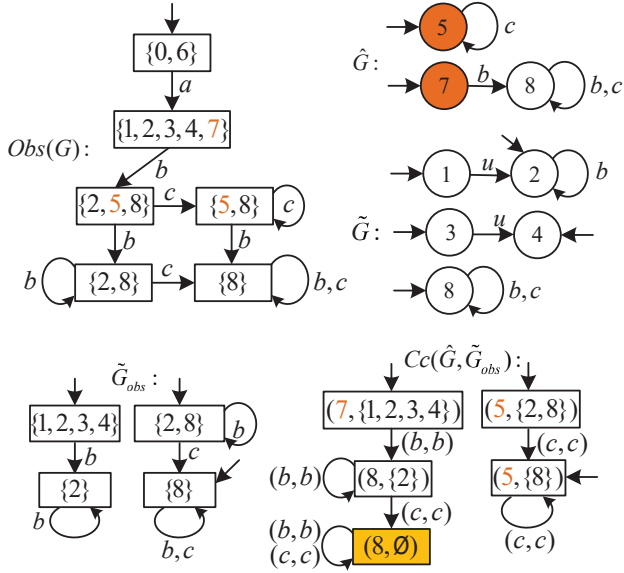}
  \caption{The constructed automata: $Obs(G)$, $\hat{G}$, $\tilde{G}$, $\tilde{G}_{obs}$, and $Cc(\hat{G},\tilde{G}_{obs})$ for the system $G$ in Fig.~\ref{Fig1}.}
  \label{Fig2}
\end{figure}
\end{example}

\subsection{An improved upper bound on $K$ in strong $K$-step opacity}\label{subsec3.4}

In~\cite{Ma(2021)}, the authors proposed a notion of strong infinite-opacity (Inf-SSO) and investigated its verification.
Recently, our previous work reduces the time complexity of~\cite{Ma(2021)} for determining Inf-SSO, see~\cite{Han(2022)} for details.
Furthermore, the authors in~\cite{Ma(2021)} shown that Inf-SSO and $(|X|(2^{|X|}-1))$-SSO are equivalent.
This indicates that $(|X|(2^{|X|}-1))$ is an upper bound on the value of $K$ in $K$-SSO.
In fact, this upper bound is conservative.
In this subsection, we propose a new upper bound on the value of $K$ in $K$-SSO by using the proposed concurrent composition $Cc(\hat{G},\tilde{G}_{obs})$, which is smaller than that in~\cite{Ma(2021)}.
And then, we establish an equivalent relationship between $K$-SSO and Inf-SSO.

Now we are ready to derive the new upper bound on the value of $K$ in $K$-SSO.
Specifically, by Definition~\ref{de:3.1}, we conclude readily that if a system $G$ is $K$-SSO w.r.t. $\Sigma_o$ and $X_S$, then it is also $K^\prime$-SSO for any $K^\prime\leq K$.
Conversely, if $G$ is not $K$-SSO with $K> |\hat{X}|2^{|X\backslash X_S|}-1$, by Definition~\ref{de:3.1},
we can know that in $G$: 1) there exists a run $x_0\stackrel{s}{\rightarrow}x_s\stackrel{t}{\rightarrow}x_t$, where $x_0\in X_0$, $x_s\in X_S$, and $P(t)=k\leq K$,
and 2) there exists no run $x^\prime_0\stackrel{s^\prime}{\rightarrow}{x^\prime_s}\stackrel{t^\prime}{\rightarrow}{x^\prime_t}$ such that ${x^\prime_s}\stackrel{t^\prime}{\rightarrow}{x^\prime_t}$ is non-secret, where $x^\prime_0\in X_0$, $P(s^\prime)=P(s)$, and $P(t^\prime)=P(t)$.
Thus, we can obtain that state $(x_t,\emptyset)$ in $Cc(\hat{G},\tilde{G}_{obs})$ is $k$-step observationally reachable from $\hat{X}_{cc,0}$ (see ``$\Leftarrow$" part in the proof of Theorem~\ref{th:3.1}).
Since $Cc(\hat{G},\tilde{G}_{obs})$ has at most $|\hat{X}|2^{|X\backslash X_S|}$ states, there exists a $k^\prime\leq|\hat{X}|2^{|X\backslash X_S|}-1$ such that $(x_t,\emptyset)$ is $k^\prime$-step observationally reachable from $\hat{X}_{cc,0}$.
By Theorem~\ref{th:3.1}, $G$ is not $k^\prime$-SSO.
Hence, it is not $(|\hat{X}|2^{|X\backslash X_S|}-1)$-SSO.
This means that the result of determining $K$-SSO using Theorem~\ref{th:3.1} does not depend on the value of $K$ when $K\geq|\hat{X}|2^{|X\backslash X_S|}-1$.
Therefore, a new upper bound on $K$ in $K$-SSO is $|\hat{X}|2^{|X\backslash X_S|}-1$, which reduces the previous upper bound of $|X|(2^{|X|}-1)$ derived in~\cite{Ma(2021)} when the size of system $G$ is relatively large.

The following two corollaries improve the corresponding results (cf., Theorem 5.2 and Corollary 5.1) in~\cite{Ma(2021)}.
Note that, we here omit their proofs, since they can be directly obtained from Definition~\ref{de:3.1} and Theorem~\ref{th:3.1}.

\begin{corollary}\label{cor:3.1}
Given a system $G=(X,\Sigma,\delta,X_0)$, a projection map $P$ w.r.t. the set $\Sigma_o$ of observable events, and a set $X_{S}$ of secret states, $G$ is $K$-SSO w.r.t. $\Sigma_o$ and $X_S$ if and only if it is min$\{K, |\hat{X}|2^{|X\backslash X_S|}-1\}$-SSO w.r.t. $\Sigma_o$ and $X_S$.
\end{corollary}

\begin{corollary}\label{cor:3.2}
Given a system $G=(X,\Sigma,\delta,X_0)$, a projection map $P$ w.r.t. the set $\Sigma_o$ of observable events, and a set $X_{S}$ of secret states, $G$ is Inf-SSO w.r.t. $\Sigma_o$ and $X_S$ if and only if it is $(|\hat{X}|2^{|X\backslash X_S|}-1)$-SSO w.r.t. $\Sigma_o$ and $X_S$.
\end{corollary}

\section{Concluding remarks}\label{sec4}

In this paper, we revisited the verification of strong $K$-step opacity for partially-observed discrete-event systems.
A new concurrent-composition structure was proposed.
Using it, we provided an improved verification algorithm for strong $K$-step opacity, which has time complexity $\mathcal{O}((|\Sigma_o||\Sigma_{uo}|+|\Sigma|)|X|^{2}2^{|X|})$
compared with time complexity $\mathcal{O}(|\Sigma_o||\Sigma_{uo}|$ $2^{|X|+|X|^2})$ (resp., $\mathcal{O}(|\Sigma_o||\Sigma_{uo}|$ $|X|^{2}2^{(K+2)|X|})$) of the previous algorithm in~\cite{Falcone(2015)} (resp., \cite{Ma(2021)}).
Furthermore, we derived a new upper bound of $|\hat{X}|2^{|X\backslash X_S|}-1$ in strong $K$-step opacity, which also reduces the previous upper bound of $|X|(2^{|X|}-1)$ given in~\cite{Ma(2021)}.

In the future, we plan to exploit the proposed concurrent-composition approach to do more efficient enforcement for strong K-step opacity compared with the enforcement algorithms obtained in~\cite{Ma(2021)}.
It would be of interest to extend the previously-proposed approach~\cite{Han(2022)} to design algorithms for enforcing strong current-state opacity, strong initial-state opacity, and strong infinite-step opacity.


\bibliographystyle{plain}

\begin{thebibliography}{99}

\bibitem{Lafortune(2018)}
S. Lafortune, F. Lin, and C.N. Hadjicostis.
On the history of diagnosability and opacity in discrete event systems.
{\it Annual Reviews in Control}, 45:257--266, 2018.

\bibitem{Hadjicostis(2020)}
C.N. Hadjicostis.
{\it Estimation and Inference in Discrete Event Systems.}
Springer, Switzerland AG, 2020.

\bibitem{An(2020)}
L. An and G. Yang.
Opacity enforcement for confidential robust control in linear cyber-physical systems.
{\it IEEE Transactions on Automatic Control}, 265(3):1234--1241, 2020.

\bibitem{Ramasubramanian(2020)}
B. Ramasubramanian, W.R. Cleaveland, and S. Marcus.
Notions of centralized and decentralized opacity in linear systems.
{\it IEEE Transactions on Automatic Control}, 265(4):1442--1455, 2020.

\bibitem{Yin(2021)}
X. Yin, M. Zamani, and S. Liu.
On approximate opacity of Cyber-physical systems.
{\it IEEE Transactions on Automatic Control}, 66(4):1630--1645, 2021.

\bibitem{Mazare(2004)}
L. Mazar$\acute{e}$.
Using unification for opacity properties.
In: Proceeding of the Workshop on Issues in the Theory of Security, pages 165--176, 2004.

\bibitem{Bryans(2005)}
J.W. Bryans, M. Koutny, and P. Ryan.
Modelling opacity using Petri nets.
{\it Electronic Notes in Theoretical Computer Science}, 121:101--115, 2005.

\bibitem{Bryans(2008)}
J.W. Bryans, M. Koutny, L. Mazar$\acute{e}$, and P. Ryan.
Opacity generalised to transition systems.
{\it Internationa Journal of Information Security}, 7(6):421--435, 2008.

\bibitem{Badouel(2007)}
E. Badouel, M. Bednarczyk, A. Borzyszkowski, et al.
Concurrent secrets.
{\it Discrete Event Dynamic Systems}, 17(4):425--446, 2007.

\bibitem{Jacob(2016)}
R. Jacob, J.J. Lesage, and J.M. Faure.
Overview of discrete event systems opacity: Models, validation, and quantification.
{\it Annual Reviews in Control}, 41:135--146, 2016.

\bibitem{Saboori(2007)}
A. Saboori and C.N. Hadjicostis.
Notions of security and opaicty in discrete event systems.
In: Proceedings of 46th IEEE Conference on Decision and Control, pages 5056--5061, 2007.

\bibitem{Saboori(2013)}
A. Saboori and C.N. Hadjicostis.
Verification of initial-state opacity in security appications of discrete event systems.
{\it Information Sciences}, 246:115--132, 2013.

\bibitem{Saboori(2011a)}
A. Saboori and C.N. Hadjicostis.
Verification of $K$-step opacity and analysis of its complexity.
{\it IEEE Transactions on Automation Science and Engineering}, 8(3):549--559, 2011.

\bibitem{Saboori(2012)}
A. Saboori and C.N. Hadjicostis.
Verification of infinite-step opacity and complexity considerations.
{\it IEEE Transactions on Automatic Control}, 57(5):1265--1269, 2012.

\bibitem{Lin(2011)}
F. Lin.
Opacity of discrete event systems and its applications.
{\it Automatica}, 47(3):496--503, 2011.

\bibitem{Wu(2013)}
Y. Wu and S. Lafortune.
Comparative analysis of related notions of opacity in centralized and coordinated architectures.
{\it Discrete Event Dynamic Systems}, 23(3):307--339, 2013.

\bibitem{Yin(2017)}
X. Yin and S. Lafortune.
A new approach for the verification of infinite-step and K-step opacity using two-way observers.
{\it Automatica}, 80:162--171, 2017.

\bibitem{Lan(2020)}
H. Lan, Y. Tong, J. Guo, and A. Giua.
Comments on ``A new approach for the verification of infinite-step and K-step opacity using two-way observers".
{\it Automatica}, 122:109290, 2020.

\bibitem{Balun(2021)}
J. Balun and T. Masopust.
Comparing the notions of opacity for discrete-event systems.
{\it Discrete Event Dynamic Systems}, 31:553--582, 2021.

\bibitem{Balun(2022)}
J. Balun and T. Masopust.
K-step opacity in discrete event systems: Verification, complexity, and relations.
{\tt\small https://arxiv.org/ abs/2109.02158}.

\bibitem{Dubreil(2010)}
J. Dubreil, P. Darondeau, and H. Marchand.
Supervisory control for opacity.
{\it IEEE Transactions on Automatic Control}, 55(5):1089--1100, 2010.

\bibitem{Saboori(2012b)}
A. Saboori and C.N. Hadjicostis.
Opacity-enforcing supervisory strategies via state estimator constructions.
{\it IEEE Transactions on Automatic Control}, 57(2):1155--1165, 2012.

\bibitem{Tong(2018)}
Y. Tong, Z. Li, C. Seatzu, and A. Giua.
Current-state opacity enforcement in discrete event systems under incomparable observations.
{\it Discrete Event Dynamic Systems}, 28(2):161--182, 2018.

\bibitem{Ji(2018)}
Y. Ji, Y. Wu, and S. Lafortune.
Enforcement of opacity by public and private insertion functions.
{\it Automatica}, 93:369--378, 2018.

\bibitem{Ji(2019a)}
Y. Ji, X. Yin, and S. Lafortune.
Opacity enforcement using nondeterministic publicly-known edit functions.
{\it IEEE Transactions on Automatic Control}, 64(10):4369--4376, 2019.

\bibitem{Ji(2019b)}
Y. Ji, X. Yin, and S. Lafortune.
Enforcing opacity by insertion functions under multiple energy constraints and imperfect information.
{\it Automatica}, 108:1--14, 2019.

\bibitem{Yin(2020)}
X. Yin and S. Li.
Synthesis of dynamic masks for infinite-step opacity.
{\it IEEE Transactions on Automatic Control}, 65(4):1429--1441, 2020.

\bibitem{Zhang(2015)}
B. Zhang, S. Shu, and F. Lin.
Maximum information release while ensuring opacity in discrete event systems.
{\it IEEE Transactions on Automation Science and Engineering}, 12(4):1067--1079, 2015.

\bibitem{Moulton(2022)}
R.H. Moulton, B.B. Hamgini, Z.A. Khouzani, et al.
Using subobservers to synthesize opacity-enforcing supervisors.
{\tt\small https://arxiv. org/abs/2110.04334}.

\bibitem{Tong(2017)}
Y. Tong, Z. Li, C. Seatzu, and A. Giua.
Verification of state-based opacity using Petri nets.
{\it IEEE Transactions on Automatic Control}, 62(6):2823--2837, 2017.

\bibitem{Zhang(2019)}
K. Zhang, X. Yin, and M. Zamani.
Opacity of nondeterministic transition systems: A (bi)simulation relation approach.
{\it IEEE Transactions on Automatic Control}, 64(2):5116--5123, 2019.

\bibitem{Keroglou(2017)}
C. Keroglou and C.N. Hadjicostis.
Probabilistic system opacity in discrete event systems.
{\it Discrete Event Dynamic Systems}, 28:289--314, 2018.

\bibitem{Yin(2019a)}
X. Yin, Z. Li, W. Wang, and C. Liu.
Infinite-step opacity and K-step opacity of stochastic discrete-event systems.
{\it Automatica}, 99:266--274, 2019.

\bibitem{Yin(2019b)}
X. Yin and S. Li.
Opacity of networked supervisory control systems over insecure multiple channel networks.
In: Proceedings of the 58th IEEE Conference on Decision and Control, pages 7641--7646, 2019.

\bibitem{Hou(2022)}
J. Hou, X. Yin, and S. Li.
A framework for current-state opacity under dynamic information release mechanism.
{\it Automatica}, accepted and in press. {\tt\small https://arxiv.org/abs/2012.04874}.

\bibitem{Saboori(2011b)}
A. Saboori and C.N. Hadjicostis.
Coverage analysis of mobile agent trajectory via state-based opacity formulations.
{\it Control Engineering Practice}, 19(9):967--977, 2011.

\bibitem{Wu(2014)}
Y. Wu, K. Sankararaman, and S. Lafortune.
Ensuring privacy in location-based services: An approach based on opacity enforcement.
In: Proceedings of 12th International Workshop on Discrete Event Systems, pages 33--38, 2014.

\bibitem{Bourouis(2017)}
A. Bourouis, K. Klai, N.B. Hadj-Alouane, and Y.E. Touati.
On the verification of opacity in web services and their composition.
{\it IEEE Transactions on Services Computing}, 10(1):66--79, 2017.

\bibitem{Lin(2020)}
F. Lin, W. Chen, W. Wang, and F. Wang.
Information control in networked discrete event systems and its application to battery management systems.
{\it Discrete Event Dynamic Systems}, 30(2):243--268, 2020.

\bibitem{Falcone(2015)}
Y. Falcone and H. Marchand.
Enforcement and validation (at runtime) of various notions of opacity.
{\it Discrete Event Dynamic Systems}, 25:531--570, 2015.

\bibitem{Ma(2021)}
Z. Ma, X. Yin, and Z. Li.
Verification and enforcement of strong infinite- and $k$-step opacity using state recognizers.
{\it Automatica}, 133:109838, 2021.

\bibitem{Han(2022)}
X. Han, K. Zhang, J. Zhang, Z. Li, and Z. Chen.
Strong current-state and initial-state opacity of discrete-event systems.
{\tt\small https:// arxiv.org/abs/2109.05475}.

\bibitem{Hopcroft(2001)}
J.E. Hopcroft, R. Motwani, and J.D. Ullman.
{\it Introduction to automata theory, languages, and computation}.
Addison-Wesley, 3nd Edition, 2001.

\bibitem{Cassandras(2008)}
C.G. Cassandras and S. Lafortune.
{\it Introduction to disctete event systems.}
Springer, New York, 2nd Edition, 2008.


\bibitem{Cormen(2009)}
T. H. Cormen, C. E. Leiserson, R. L. Rivest, and C. Stein.
{\it Introduction to Algorithms}.
MIT Press, 2009.

\end{thebibliography}

\end{document}